%% file: main.tex
\documentclass[11pt]{article}
\usepackage{geometry}

\usepackage[english]{babel}
\usepackage{amssymb,amsmath}
\usepackage{array}
\usepackage{amsthm}
\usepackage[symbol]{footmisc}
\usepackage{graphicx}
\usepackage[shortlabels]{enumitem}
\usepackage[dvipsnames]{xcolor} 
\usepackage{algorithm}
\usepackage{algpseudocodex}
\usepackage{authblk}
\usepackage{mathtools}

\makeatletter
\newcommand{\oset}[3][0ex]{%
  \mathrel{\mathop{#3}\limits^{
    \vbox to#1{\kern-1\ex@
    \hbox{$\scriptstyle#2$}\vss}}}}
\makeatother

\algrenewcommand\algorithmicrequire{\textbf{Given:}}
\algrenewcommand\algorithmicensure{\textbf{Output:}}
\algrenewcommand\algorithmicforall{\textbf{for each}}
\def\lastpage@putlabel{}

\usepackage[colorlinks=true, allcolors=blue]{hyperref}

\newcommand{\brak}[1]{\left\langle #1\right\rangle}

\newcommand{\abs}[1]{\left\lvert #1\right\rvert}

\newcommand{\infnorm}[1]{\Vert #1\Vert_\infty}
\renewcommand{\Bar}[1]{\overline{#1}}

\newcommand{\eps}{\varepsilon}

\newcommand{\Number}[1]{\mathbb{#1}}

\newcommand{\Filtration}[1]{\mathcal{#1}}
\newcommand{\Text}[1]{\text{#1}}

\newcommand{\R}{\Number{R}}

\newcommand{\Z}{\Number{Z}}

\renewcommand{\S}{\Number{S}}

\newcommand{\G}{\Filtration{G}}
\newcommand{\cG}{\mathcal{G}}

\newcommand{\wt}{\Text{c}}

\newcommand{\Ghat}{\hat{G}}

\newcommand{\tb}{{t^-}}
\newcommand{\ta}{{t^+}}

\newcommand{\uheap}{\mathcal{H}_U}
\newcommand{\lheap}{\mathcal{H}_L}

\DeclareMathOperator{\aug}{\mathrm{Aug}}
\DeclareMathOperator{\Aug}{\aug}
\DeclareMathOperator{\ext}{\mathrm{ext}}
\DeclareMathOperator{\PHT}{\mathrm{PHT}}
\DeclareMathOperator{\Dgm}{\mathrm{Dgm}}

\newtheorem*{theorem*}{Theorem}
\newtheorem{theorem}{Theorem}[section]
\newtheorem{proposition}[theorem]{Proposition}

\newtheorem{lemma}[theorem]{Lemma}

\theoremstyle{definition}
\newtheorem{definition}[theorem]{Definition}
\newtheorem{remark}[theorem]{Remark}

\title{\LARGE The Kinetic Hourglass  Data Structure for Computing the Bottleneck Distance of Dynamic Data}
\author[1, 2]{\large Elizabeth Munch}
\author[3]{\large Elena Xinyi Wang}
\author[4]{\large Carola Wenk}

\affil[1]{\footnotesize Department of Computational Mathematics, Science, and Engineering, Michigan State University}
\affil[2]{\footnotesize Department of Mathematics, Michigan State University}
\affil[3]{\footnotesize Department of Informatics, University of Fribourg}
\affil[4]{\footnotesize Department of Computer Science, Tulane University}

\begin{document}
\date{}
\maketitle

\begin{abstract}
The kinetic data structure (KDS) framework is a powerful tool for maintaining various geometric configurations of continuously moving objects. 
In this work, we introduce the kinetic \emph{hourglass}, a novel KDS implementation designed to compute the bottleneck distance for geometric matching problems. 
We detail the events and updates required for handling general graphs, accompanied by a complexity analysis. 
Furthermore, we demonstrate the utility of the kinetic hourglass by applying it to compute the bottleneck distance between two persistent homology transforms (PHTs) derived from shapes in $\mathbb{R}^2$, which are topological summaries obtained by computing persistent homology from every direction in $\mathbb{S}^1$.
\end{abstract}

\input{content_noAppendix}

\bibliographystyle{plain}
\bibliography{hourglassBib}

\end{document}

%% file: content_noAppendix.tex
\section{Introduction}
Motion is a fundamental property of the physical world. 
To address this challenge computationally, Basch et al.~proposed the \textit{kinetic data structure} (KDS) \cite{BaschGuibasHeapOG1997}, designed to maintain various geometric configurations for continuously moving objects. 
The KDS frameworks have been applied to many geometric problems since it was introduced, including but not limited to finding the convex hull of a set of moving points in the plane \cite{GuibasKDS1999}, the closest pair of such a set \cite{GuibasKDS1999}, a point in the center region \cite{Agarwal2005Center}, kinetic medians and $kd$-trees \cite{Agarwal2002kdTree}, and range searching; see \cite{GuibasSurvay} for a survey.
In this work, we extend the framework to the geometric matching problem.
Specifically, we are interested in the min-cost matching of a weighted graph with continuously changing weights on the edges. 

This is related to the problem of comparing \textit{persistence diagrams}, a central concept in topological data analysis (TDA).
Intuitively, a persistence diagram is a visual summary of the topological features of a given object.
We can learn how different or similar two objects are by comparing their persistence diagrams.
One of the most common comparisons is to compute the bottleneck distance, a dual of the min-cost matching problem in graphs.
Improving the algorithm to compute the bottleneck distance has been studied extensively, both theoretically and practically \cite{Dey2018Bottleneck, HopcroftKarp1973,Efrat2001,KerberGeometryHelps2017, Cabello2024Matching}.

\emph{Vineyards} are continuous families of persistence diagrams that represent persistence for time-series of continuous functions \cite{Morozov2006Vineyard, Xian2022Crocker, Kim2021}.
A specific case of vineyards is the persistent homology transform (PHT) \cite{CurryPHT2018, GhristPHT2018,TurnerPHT2014}, which is the family of persistence diagrams of a shape in $\R^d$ computed from every direction in $\S^{d-1}$.
The PHT possesses desirable properties such as continuity,  injectivity, and stability. 
However, an efficient method for comparing two PHTs has not been thoroughly investigated. 
In particular, existing approaches to compute the bottleneck distance between two PHTs of shapes in $\R^2$ rely on sampling various directions, which only provides approximate results. 
To date, there is no known solution for computing the exact bottleneck distance. 
This work presents the first method to address this challenge, providing a precise and efficient solution.

\textbf{Our contribution}: We construct a new kinetic data structure that maintains a min-cost matching and the bottleneck distance of a weighted bipartite graph with continuously changing weight functions. 
We evaluate the data structure based on the four standard metrics in the KDS framework.
We provide a specific case study where the weighted graph is computed from two persistent homology transforms, resulting in the first exact distance computation between two PHTs.

The paper is organized as follows: we cover the bottleneck matching problem's background and the kinetic data structures' preliminaries in Section \ref{Sec:Background}. 
In Section \ref{Sec:KDS}, we iterate the events and updates necessary to maintain the KDS and the overall complexity analysis. 
Finally, we focus in Section \ref{Sec:PHT} on our case study of the persistent homology transform.

\section{Background}
\label{Sec:Background}
Broadly, we are interested in a geometric matching problem. 
Given an undirected graph $G = (V, E)$, a \emph{matching} $M$ of $G$ is a subset of the edges $E$ such that no vertex in $V$ is incident to more than one edge in $M$. 
A vertex $v$ is said to be \emph{matched} if there is an edge $e\in M$ that is incident to $v$.
A matching is \emph{maximal} if it is not properly contained in any other matching. 
A \emph{maximum} matching $M$ is the matching with the largest cardinality; i.e., for any other matching $M'$, $\abs{M} \geq \abs{M'}$.
A maximum matching is always maximal; the reverse is not true. 

For a graph $G = (X\sqcup Y, E)$ where $\abs{X}=\abs{Y} = n$ and $\abs{E} = m$, a maximum matching is a \emph{perfect} matching if every $v\in X\sqcup Y$ is matched, and $\abs{M} = n$.
This can be expressed as a bijection $\eta:X\rightarrow Y$.
For a subset $W\subseteq X$, let $N(W)$ denote the neighborhood of $W$ in $G$, the set of vertices in $Y$ that are adjacent to at least one vertex of $W$.
Hall's marriage theorem provides a necessary condition for a bipartite graph to have a perfect matching.

\begin{theorem}[Hall's Marriage Theorem]
\label{thm:Hall}
    A bipartite graph $G = (X\sqcup Y, E)$ has a perfect matching if and only if for every subset $W$ of $X$: $\abs{W}\leq \abs{N(W)}.$
\end{theorem}

Building on Hall's Marriage Theorem, we extend our focus to weighted bipartite graphs. 
Given such a graph, a fundamental optimization problem is to identify matchings that minimize the maximum edge weight, known as the \emph{bottleneck cost}.

\begin{definition}
    
A weighted graph $\mathcal{G} = (G, \wt)$ is a graph $G$ together with a weight function $\wt: E\rightarrow\R_+$. 
The \emph{bottleneck cost} of a matching $M$ for such a $\mathcal{G}$ is $\max\{\wt(e)\mid e\in M\}$.
The \emph{bottleneck edge} is the highest weighted edge in $M$, assuming this is unique.
A perfect matching is \emph{optimal} if its cost is minimal among all perfect matchings.
An optimal matching is also called a \emph{min-cost} matching.
\end{definition}

To find a maximum matching of a graph, we use augmenting paths.

\begin{definition}
	For a graph $G$ and matching $M$, a path $P$ is an \emph{augmenting path} for $M$ if:
	\begin{enumerate}
		\item the two end points of $P$ are unmatched in $M$, and
		\item the edges of $P$ alternate between edges $e\in M$ and $e\notin M$.
	\end{enumerate}
\end{definition}

\begin{theorem}[Berge's Theorem]
\label{thm:Berge}
    A matching $M$ in a graph $G$ is a maximum matching if and only if $G$ contains no $M$-augmenting path.
\end{theorem}

The existing algorithms that compute the bottleneck cost are derived from the Hopcroft-Karp maximum matching algorithm~\cite{HopcroftKarp1973}, which we briefly review.
Given a graph $G = (X\sqcup Y, E)$ where $\abs{E} = n$, and an initial matching $M$, the algorithm iteratively searches for augmenting paths $P$.
Each phase begins with a breadth-first search from all unmatched vertices in $X$, creating a layered subgraph of $G$ by alternating between $e\in M$ and $e\notin M$.
It stops when an unmatched vertex in $Y$ is reached.
From this layered graph, we can find a maximal set of vertex-disjoint augmenting paths of the shortest length.
For each $P$, we augment $M$ by replacing $M\cap P$ with $P\backslash M$.
We denote this process as $\Aug(M, P) = M\setminus(M\cap P)\cup(P\setminus M)$.
Note that $\abs{\aug(M,P)} = \abs{M}+1$, so we can  repeat the above process until no more augmenting paths can be found.
By Theorem \ref{thm:Berge}, the resulting $M$ is maximum.
The algorithm terminates in $O(\sqrt{n})$ rounds, resulting in a total running time of $O(n^{2.5}).$
This algorithm was later improved to $O(n^{1.5}\log n)$ for geometric graphs by Efrat et al.\ by constructing the layered graph via a near-neighbor search data structure \cite{Efrat2001}.

\subsection{Bottleneck Distance}
\label{sec:BackgroundPersistence}
    Let $X$ and $Y$ be two sets of $n$ points. 
    We consider the bipartite graph obtained by adding edges between points whose weight $\wt: E \to \R_{\geq0}$ is at most $\lambda$:
    \[
    G_\lambda = (X\sqcup Y, \{xy\mid \wt(xy) \leq \lambda\}).
    \]

    \begin{definition}
    The \emph{optimal bottleneck cost} is 
    the minimum $\lambda$ such that $G_\lambda$ has a perfect matching, denoted as $d_B(G)$. 
    \end{definition}    

    We are interested in computing the bottleneck distance in the context of \emph{persistent homology}.
    Persistent homology is a multi-scale summary of the connectivity of objects in a nested sequence of subspaces; see \cite{DeyWang2017} for an introduction. 
    For the purposes of this section, we can define a 
    \emph{persistence diagram} to be a finite collection of points $\{(b_i,d_i)\}_i$ with $d_i \geq b_i$ for all $i$. 
    Further details and connections to the input data will be given in Section \ref{Sec:PHT}.
    
    Given two persistence diagrams $X$ and $Y$, a partial matching is a bijection $\eta:X' \to Y'$ on a subset of the points $X' \subseteq X$ and $Y' \subseteq Y$; we denote this by $\eta: X \rightleftharpoons Y$. 
    The cost of a partial matching is the maximum over the $L_\infty$-norms of all pairs of matched points and the distance between the unmatched points to the diagonal:
    \begin{equation*}
        c(\eta) = \max \left( 
        \{ \|x-\eta(x)\|_\infty \mid x \in X'\} 
        \cup
        \{ \tfrac{1}{2}|z_2-z_1|  \mid 
        (z_1,z_2) \in (X \setminus X') \cup (Y \setminus Y') \}
        \right)
    \end{equation*}
    and 
    the bottleneck distance is defined as
    $
    d_B(X, Y) = \inf_{\eta:X\rightleftharpoons Y} c(\eta)
    $.

    We now reduce finding the bottleneck distance between persistence diagrams to a problem of finding the bottleneck cost of a bipartite graph.
    Let $X$ and $Y$ be two persistence diagrams given as finite lists of off-diagonal points. 
    For any off-diagonal point $z = (z_1, z_2)$, the orthogonal projection to the diagonal is 
    $z' = ((z_1+z_2)/2, (z_1+z_2)/2)$. 
    Let $\Bar{X}$ (resp.~$\Bar{Y}$) be the set of orthogonal projections of the points in $X$ (resp.~$Y$). 
    Set $U = X\sqcup \Bar{Y}$ and $V = Y\sqcup \Bar{X}$.
    We define the complete bipartite graph $G = (U\sqcup V, U\times V, \wt)$, where for $u\in U$ and $v\in V$, the weight function $\wt$ is given by
    \[
    \wt(uv) = 
    \begin{cases}
        \infnorm{u - v} &\text{if $u\in X$ or $v\in Y$}\\
        0 &\text{if $u\in \Bar{X}$ and $v\in \Bar{Y}$}.
    \end{cases}
    \]
    An example of the bipartite graph construction is shown in Figure \ref{fig:persBottleneck}.
    This graph can be used to compute the bottleneck distance of the input diagrams because of the following lemma.
    \begin{lemma}[Reduction Lemma \cite{EdelsHarer2010}]
        For the above construction of $G$,  $d_B(G) = d_B(X, Y)$.
    \end{lemma}

Naively, given a graph $G$ of size $n$, we can compute the bottleneck distance by sorting the edge weights and performing a binary search for the smallest $\lambda$ such that $G_\lambda$ has a perfect matching in $O(n^2\log n)$ time, which dominates the improved Hopcroft-Karp algorithm. 
Using the technique for efficient $k$-th distance selection for a bi-chromatic point set under the $L_\infty$ distance introduced by Chew and Kedem, this can be reduced to $O(n^{1.5}\log n)$ \cite{ChewKedem1992}.
Therefore, the overall complexity to compute the bottleneck distance of a static pair of persistence diagrams is $O(n^{1.5}\log n)$.

\begin{figure}
    \centering
    \includegraphics[width=0.6\linewidth]{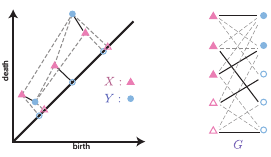}
    \caption{Construction of the bipartite graph $G$ based on the persistence diagrams $X$ and $Y$. 
    }
    \label{fig:persBottleneck}
\end{figure}
    
\subsection{Kinetic Data Structure}
    A kinetic data structure (KDS) \cite{GuibasKDS1998STA, GuibasKDS1999, GuibasSurvay} maintains a system of objects $v$ that move along a known continuous \emph{flight plan} as a function of time, denoted as $v = f(v)$. 
    \emph{Certificates} are conditions under which the data structure is accurate, and \emph{events} track when the certificate fails and what corresponding \emph{updates} need to be made.
    The certificates of the KDS form a proof of correctness of the current configuration function at all times. 
    Updates are stored in a priority queue, keyed by event time.
    Finally, to advance to time $t$, we process all the updates keyed at times before $t$ and pop them from the queue after updating.
    We continue until $t$ is smaller than the first event time in the queue.

    A KDS is evaluated by four measures. 
    We say a quantity is \emph{small} if it is a polylogarithmic function of $n$, or is $O(n^\eps)$ for arbitrarily small $\eps$, and $n$ is the number of objects. 
    A KDS is considered \emph{responsive} if the worst-case time required to fix the data structure and augment a failed certificate is small.
    \textit{Locality} refers to the maximum number of certificates in which any one value is involved.
    A KDS is local if this number is small.
    A \emph{compact} KDS is one where the maximum number of certificates used to augment the data structure at any time is $O(n \text{ polylog }n)$ or $O(n^{1+\eps}).$
    Finally, we distinguish between \emph{external events}, i.e., those affecting the configuration function (e.g., maximum or minimum values) we are maintaining, and \emph{internal events}, i.e., those that do not affect the outcome. 
    We consider the ratio between the worst-case total number of events that can occur when the structure is advanced to $t=\infty$ and the worst-case number of external to the data structure. 
    The second number is problem-dependent. 
    A KDS is \emph{efficient} if this ratio is small.

    We next explore kinetic solutions to upper- and lower-envelope problems. The deterministic algorithm uses a kinetic heap, while the more efficient kinetic hanger adds randomness.
    
\subsubsection{Kinetic Heap}
    The kinetization of a heap results in a \emph{kinetic heap}, which maintains the priority of a set of objects as they change continuously. 
    Maximum and minimum are both examples of priorities. 
    A kinetic heap follows the properties of a static heap such that objects are stored as a balanced tree.
    If $v$ is a child node of $u$, then $u$ has higher priority than $v$. 
    In the example of a max heap, that means $u>v$.
    In a kinetic max heap, the value of a node is stored as a function of time $f_X(t)$.
    The data structure augments a certificate $[A>B]$ for every pair of parent-child nodes $A$ and $B$.
    It is only valid when $f_A(t)>f_B(t).$
    Thus, the failure times of the certificate are scheduled in the event queue at time $t$ such that $f_A(t)=f_B(t).$
    When a certificate $[A>B]$ fails, we swap $A$ and $B$ in the heap and make $5$ parent-child updates.
    This is the total number of certificates in which any pair of $A$ and $B$ are present.
    The kinetic max heap supports operations similar to a static heap, such as create-heap, find-max, insert, and delete. 
    Insertion and deletion are done in $O(\log n)$ time.

    This KDS satisfies the responsiveness criteria because, in the worst case, each certificate failure only leads to $O(1)$ updates. 
    Locality, compactness, and efficiency depend on the behavior of $f_X(t)$.
    The analysis below assumes the case of affine motion where $f_X(t) = at +b$.
    In the worst case, each node is present in only three certificates, one with its parent and two with children, resulting in $3$ events per node. 
    For compactness, the total number of scheduled events is $n-1$ for $n$ nodes in the kinetic heap. 
    It is also efficient, where the maximum number of events processed is $O(n\log n)$ for a tree of height $O(\log n)$.
    The total time complexity in the linear case is $O(n\log^2 n)$ \cite{BaschGuibasHeapOG1997}.
    The results for locality and compactness hold for the more general setting of $s$-intersecting curves, where each pair of curves intersect at most $s$ times, $s\in\Z$. 
    The efficiency is unknown in this case, and the time complexity is $O(n^2\log^2 n)$ \cite{BaschGuibas2003Curves}.

    When the functions are only defined within an interval of time, we call %
    them segments.
    We perform an insertion when a new segment appears and a deletion when a segment disappears.
    In the scenario of affine motion segments, the complexity of the kinetic heap is $O(n\sqrt{m}\log^{3/2}m)$, $n$ is the total number of elements and $m$ is the maximum number of elements stored at a given time.
    In the case of $s$-intersecting curve segments, the complexity of the kinetic heap is $O(mn\log^2 m)$.
    
\begin{figure}
    \centering
    \includegraphics[width=0.85\linewidth]{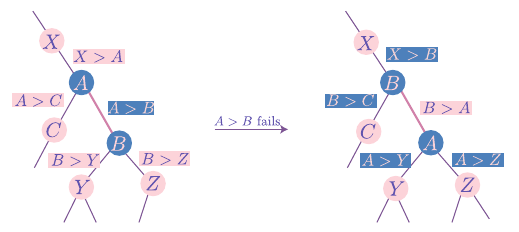}
    \caption{Kinetic heap event updates}
    \label{fig:kineticHeap}
\end{figure}

\subsubsection{Kinetic Hanger}
    More efficient solutions to this problem have been developed, such as the kinetic \emph{heater} and \emph{tournament} \cite{BaschGuibas2003Curves, GuibasKDS1999}.
    While both improve some aspects of the heap, they each have downsides.
    The kinetic heater increases the space complexity and is difficult to implement, and the tournament has a high locality bound.
    We instead use the \emph{kinetic hanger}, introduced by da Fopnseca et al.\ in \cite{daFonseca2004Hanger}.
    It modifies the kinetic heap by incorporating randomization to balance the tree, meaning all complexity results are probabilistic rather than deterministic.
    
    Given an initial set of elements $S$, sorted by decreasing priorities, we construct a hanger $hanger(S)$ by so-called \emph{hanging} each element at the root.
    To hang an element $e$ at node $v$, if no element is assigned to $v$, we assign $e$ to $v$ and return.
    Otherwise, we use a random seed $r$ to choose a child $c_r$ of $v$ and recursively hang $e$ at $c_r$.
    Insertion is similar to hanging
    - it only requires the additional step of comparing the priorities of $e$ and $e'$, where $e'$ is the element already assigned to $v$.
    Deletion can be done by removing the desired element and going down the tree, replacing the current element with its child with the highest priority.

    As shown in \cite{daFonseca2004Hanger}, the kinetic hanger has $O(1)$ locality.
    The number of expected events in the affine motion case is $O(n^2\log n)$ and $O(\lambda_s(n)\log n)$ for $n$ $s$-intersecting curves.
    The expected runtime is obtained by multiplying by $O(\log n)$ for each event.
    We use $\alpha(\cdot)$ to denote the inverse Ackerman function and
    $\lambda_s$ is the length bound for the Davenport-Schinzel Sequence; see \cite{Agarwal2000} for details.
    When the functions are partially defined, the complexities are $O(n\alpha(m)\log^2m)$ for line segments and $O(n/m\lambda_{s+2}(m)\log^2 m)$ for curve segments.

\begin{table}
\centering
    \begin{tabular}{c|c|c}
        \textbf{Scenario} & \textbf{Kinetic Heap} & \textbf{Kinetic Hanger} \\
        \hline
        Lines & $O(n\log^2 n)$ & $O(n\log^2 n)$\\
        \hline
        Line segments & $O(n\sqrt{m}\log^{3/2}m)$ & $O(n\alpha(m)\log^2m)$\\
        \hline
        $s$-intersecting curves & $O(n^2\log^2 n)$ & $O(\lambda_s(n)\log^2 n)$\\
        \hline
        $s$-intersecting curve segments  & $O(mn\log^2 m)$ & $O(n/m\lambda_{s+2}(m)\log^2 m)$
    \end{tabular}
    \caption{Deterministic complexity of the kinetic heap and expected complexity of the kinetic hanger. Here, $n$ is the total number of elements, $m$ denotes the maximum number of elements stored at a given time, $s$ is a constant.}
    \label{tbl:complexity}
\end{table}

\section{Kinetic Hourglass}
\label{Sec:KDS}
In this section, we introduce a new kinetic data structure that keeps track of the optimal bottleneck cost of a weighted graph $\mathcal{G} = (G, \wt)$ where $\wt$ changes continuously with respect to time $t$.
The \emph{kinetic hourglass} data structure is composed of two kinetic heaps; in Section \ref{sec:hanger} we will give the details for replacing these with kinetic hangers. 
One heap maintains minimum priority, and the other maintains maximum. 
Assume we are given a connected bipartite graph $G = (V,E)$ with the vertex set $V = X\sqcup Y$, where $\abs{X} = \abs{Y} = n$; and edge set $E$, where $\abs{E} = m$.
If $G$ is a complete bipartite graph, then $m = n^2$.
The weight of the edges at time $t$ is given by $\wt^t:E \to \R_{\geq 0}$.
Denote the weighted bipartite graph by $\G^t = (G, \wt^t)$. 
The weights of those $m$ edges are the objects we keep track of in our \emph{kinetic hourglass}.
We assume that these weights, called flight plans in the kinetic data structure setting, are given for all times $t \in [0,T]$. 

Let $G_\delta^t \subseteq G$ be the portion of the complete bipartite graph with all edges with weight at most $\delta$ at time $t$; i.e., $V(G_\delta^t) = V$, and $E(G_\delta^t) = \{ e \in E \mid \wt^t(e) \leq \delta \}$. 
If the bottleneck distance $\hat{\delta}^t=d_B(\cG^t)$ is known, we are focused on the bipartite graph $G_{\hat{\delta}}^t$, which we will denote by $\Ghat^t$ for brevity. 
By definition, we know that there is a perfect matching in $\Ghat^t$, which we denote as $M^t$, although we note that this is not unique. 
Further, there is an edge $\hat{e}^t \in M^t$ with $\wt(\hat{e}^t) = \hat{\delta}^t$, which we call the \emph{bottleneck edge}. 
This edge is unique as long as all edges have unique weights.
We separate the remaining edges into the sets  
\begin{align*}
    L^t &= E(\Ghat^t)\setminus E(M^t), \text{ and }\\
    U^t &= E\setminus E(\Ghat^t) = \{ e \in E \mid \wt^t(e) >\hat{\delta}^t\}
\end{align*}
so that $E = L^t \sqcup M^t \sqcup U^t $.

The kinetic hourglass consists of the following kinetic max heap and kinetic min heap.
The {\em lower heap} $\lheap$ is the max heap containing $L^t \sqcup M^t = E(\Ghat^t)$.
The {\em upper heap} $\uheap$ is the min heap containing $U^t \cup \{ \hat {e}\}$.
Note that 
$\wt^t(\hat{e}) =  \max \{\wt^t(e) \mid e \in L^t \sqcup M^t  \}$ and 
$\wt^t(\hat{e}) <  \min \{\wt^t(e) \mid e \in U^t  \}$, meaning that $\hat{e}$ is the root for both heaps.

\begin{figure}
    \centering
    \includegraphics[width=\linewidth]{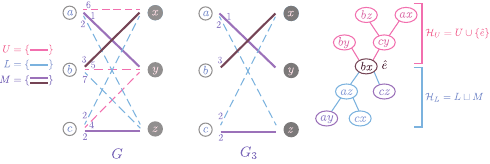}
    \caption{Illustration of construction of the kinetic hourglass.}
    \label{fig:Hourglass}
\end{figure}

\subsection{Certificates}

The certificates for the kinetic hourglass (i.e., properties held by the data structure for all $t\in[0,T]$) are
\begin{enumerate}
    \item 
All max-heap certificates for $\mathcal{H}_L$ and min-heap certificates for $\mathcal{H}_U$.

    \item 
Both heaps have the same root, denoted $r^t$.
    \item The edge $r^t$ is the edge with bottleneck cost; i.e.,
$r^t = \hat{e}^t$ where $c^t(\hat{e}^t) = \hat{\delta}^t$.
\end{enumerate}

Assuming the certificates are maintained, $r^t$ and $\hat e^t$ are the same edge. 
However, in the course of proofs, bottleneck edge of the matching is denoted by $\hat e^t$, while we use $r^t$  for the edge stored in the root of the two heaps (or $r^t(\uheap)$ and $r^t(\lheap)$ when a distinction is needed).

\subsection{Events}

For a particular event time $t$, we denote the moment of time just before an event by $t^- = t - \varepsilon$ and the moment of time just after by $t^+ = t + \varepsilon$.
For two edges, we write $a \preccurlyeq_t b$ to mean that  $\wt^t(a)\leq \wt^t(b)$. 
The two heaps have their own certificates, events (both internal and external), and updates as in the standard setting.
Our main task of this section is to determine which events in the heaps lead to an external event in the hourglass. 
We define the external events for the hourglass as those events in $\uheap$ and $\lheap$ which lead to changes of the root, $r^t$.
Thus, internal events are those which do not affect the roots; i.e.,
$r^\tb{(\uheap)} = r^\ta{(\uheap)}$ and $r^\tb{(\lheap)} = r^\ta{(\lheap)}$.

We first show that an internal event of $\uheap$ or $\lheap$ is an internal event of the hourglass. 

\begin{lemma}
\label{lem:Hall}
    If the event at time $t$ is an internal event of $\uheap$ or $\lheap$ and the kinetic hourglass satisfies all certificates at time $\tb$,
    then the edge giving the bottleneck distance for times $\tb$ and $\ta$ is the same; that is, $\hat{e}^\tb = \hat{e}^\ta$.
\end{lemma}
\begin{proof}
    Following the previous notation, we have bottleneck distances before and after given by 
    $\hat{\delta}^\tb = \wt^\tb(\hat{e}^\tb)$
    and
    $\hat{\delta}^\ta = \wt^\ta(\hat{e}^\ta)$.
    Because we start with a correct hourglass, we know that 
    $\hat{e}^\tb = r^\tb(\uheap) = r^\tb(\lheap) = r^\tb$.
    By definition, an internal event in either heap is a swap of two elements with a parent-child relationship but for which neither is the root so the roots remain unchanged; 
    that is, $r^\tb = r^\ta$ so we denote it by $r$ for brevity. 
    This additionally means that no edge moves from one heap to the other, so the set of elements in each heap does not change and thus $G_{\wt^\tb(r)}= G_{\wt^\ta(r)} $.
    Again for brevity, we write this subgraph as $\Gamma$.

    We need to show that this edge $r$ is the one giving the bottleneck distance at $\ta$, i.e.~$\hat{\delta}^\ta = \wt^\ta(r)$ or equivalently that $r = \hat e^\ta$. 
    All edges of the perfect matching from $\tb$, $M^\tb$, are contained in $\Gamma $; thus  $M^\tb$ is still a perfect matching at time $\ta$.
    We show that any minimal cost perfect matching for $\ta$ must contain $r$, and thus $M^\ta$ is a \emph{minimal} cost perfect matching. 
    Since $r$ is $\hat{e}^\tb$, by removing $r$, $\Gamma \backslash \{r\}$ will cease to have a perfect matching, else contracting the minimality of the perfect matching at $\tb$.
    But as the order is unchanged, this further means that for time $\ta$, lowering the threshold for the subgraph $\Gamma =  G_{\wt^\ta(r)}$ or equivalently removing the edge $r$ will not have a perfect matching, finishing the proof. 
\end{proof}

The remaining cases to consider are external events of $\uheap$ and $\lheap$, when a certificate in one of the heaps involving the root fails. 
This can be summarized in the three cases below. 
In each case, denote the root at time $\tb$ as $r = r^\tb$.

\begin{enumerate}
    \item ($L$-Event) Swap priority of $r$ and an $e\in L^t$ in $\lheap$; i.e.~$e \preccurlyeq_{\tb} r$ and $ r \preccurlyeq_{\ta} e$.  
    \item ($M$-Event) Swap priority of $r$ and an $e\in M^t$ in $\lheap$; i.e.~$e \preccurlyeq_{\tb} r$ and $ r \preccurlyeq_{\ta} e$. 
    \item ($U$-Event) Swap priority of $r$ and an $e\in U^t$ in $\uheap$; i.e.~$ r \preccurlyeq_{\tb} e$ and $e \preccurlyeq_{\ta} r$.
\end{enumerate}

In the remainder of this section, we consider each of those events and provide the necessary updates.
The simplest update comes from the first event in the list, since we will show that no additional checks are needed. 

\begin{lemma}[$L$-Event]
Assume we swap priority of $r=r^\tb$ and an $e\in L^\tb$ at $t$; i.e.~$e \preccurlyeq_{\tb} r$ and $ r \preccurlyeq_{\ta} e$. 
Then $e$ moves from $\lheap $ to $\uheap$, and $r$ remains the root and is the edge with the bottleneck cost for $\ta$, i.e., $r^\ta = \hat e^\ta$.
\end{lemma}
\begin{proof}
Denote $M=M^\tb$. In this case, $e$ is an edge in the lower heap but $e \not \in  M$. 
Because $r \preccurlyeq_\ta e$, in order to maintain the heap certificates, either $e$ needs to be inserted into the upper heap with $r$ remaining as the root; or if $e$ remains in the lower heap, it needs to become the new root. 
Note that although they swap orders, the graph thresholded at the cost of $r$ at $\tb$ and the graph thresholded at the cost of $e$ at $\ta$ are the same; i.e.~$G_{c^\tb(r)} = G_{c^\ta(e)}$. 
In addition, $G_{c^\ta(r)} = G_{c^\ta(e)} \setminus \{e\}$.
However, since $e\notin M$, $M$ is still a perfect matching in $G_{c^+(r)}$.
If there exists a perfect matching in $G_{c^\ta(r)} \setminus \{r\}$, this would constitute a perfect matching at time $\tb$ which has lower cost than $M$, contradicting the assumption that $M$ is a minimal cost matching for that time. 
Thus $M$ is a minimal cost matching for time $\ta$. 
\end{proof}

\begin{figure}
    \centering
    \includegraphics[width=0.7\linewidth]{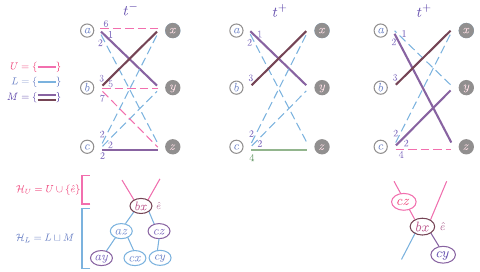}
    \caption{Illustration of Scenario 1 of and $M$-Event; see Lemma \ref{lem:M}.}
    \label{fig:Event2-1}
    \includegraphics[width=0.7\linewidth]{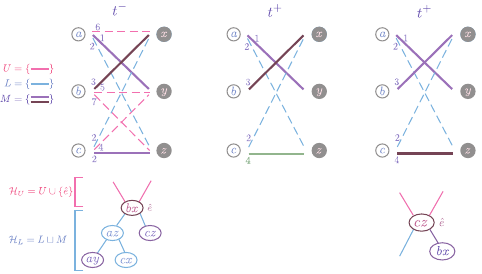}
    \caption{Illustration of Scenario 2 of $M$-Event; see Lemma \ref{lem:M}.}
    \label{fig:Event2}
\end{figure}

\begin{lemma}[$M$-Event]
\label{lem:M}
Assume we swap priority of $r=r^\tb$ and an $e\in M^\tb$ at $t$; i.e.~$e \preccurlyeq_{\ta} r$ and $ r \preccurlyeq_{\tb} e$. 
Let $G' = \hat{G}^\tb\setminus \{e\}$ be the graph with $e = (u,v)$ removed.
Exactly one of the following scenarios happens. 
(See Figure \ref{fig:Event2-1} and \ref{fig:Event2} for examples of the two scenarios.)
\begin{enumerate}
    \item There exists an augmenting path $P$ in $G'$ from $u$ to $v$. 
    Then $e$ moves into the upper heap ($e \in U^{\ta}$), 
    the root remains the same ($\hat{e}^\ta = r^{\ta} =r$), and the matching is updated with the augmenting path, specifically 
    $
    M^\ta = \Aug(M^\tb, P)
    $.
    \item There is no such augmenting path. Then $M^{\ta} = M^{\tb}$ and  $\hat{e}^\ta = r^{\ta} = e$; i.e.~the only update is that $r$ and $e$ switch places in the lower heap.
\end{enumerate}
\end{lemma}

\begin{proof}

    Let $M' = M^\tb\setminus \{e\}$. 
    Then $M'$ is a matching of size $n-1$ in $G'$ (hence it is not perfect), and the unmatched vertices are $u$ and $v$.

    If there exists an augmenting path $P$ from $u$ to $v$, augment $M'$ by replacing  $M'\cap P$ with $P\setminus M'$ to make a new matching $M''$.
    This increases $\abs{M'}$ by $1$ and thus $M''$ is a perfect matching. 
    Note that if $r \in P$, then $M''$ is also a perfect matching in $\tb$ with strictly cheaper cost than $M^{\tb}$; but this contradicts the assumption of $r$ giving the bottleneck distance. 
    Thus we can assume that $r \in M''$, and all edges in the lower heap have cost at most that of $r$.
    This means that $r$ is still the bottleneck edge of the matching, i.e.~$\wt^\ta(r) = \max\{\wt^\ta(e)\mid e\in M''\}$.
   
    Assume instead that there is no augmenting path in $G'$ from $u$ to $v$. 
    Then $M'$ is a maximum matching for $G'$ by Theorem \ref{thm:Berge}, however it is not perfect. 
    Note that because $G' = \hat{G}^\tb\setminus e$ and $e$ and $r$ have swapped places, we have that $G' = G_{\wt^\ta(r)}$.    
    Therefore, there is no perfect matching for $G_{\wt^\ta(r)}$.
    However, $M^\tb$ is a perfect matching at time $\tb$ and $\hat{G}^\tb = G_{\wt^\ta(e)}$, it still is a perfect matching at $\ta$ for $G_{\wt^\ta(e)}$.
    Moreover, $\wt^\ta(e) = \max\{\wt^\ta(e)\mid e\in M^\ta\}$.
    Therefore, $e = \hat{e}^\ta = r^\ta$ becomes the root, and $r$ becomes a child of the root $e$ in $\lheap$.
\end{proof}

\begin{lemma}[$U$-Event]
\label{lem:U}
Assume we swap priority of $r=r^\tb$ and an $e\in U^\tb$ at $t$; i.e.~$ r \preccurlyeq_{\tb} e$ and $e \preccurlyeq_{\ta} r$.
Let $G' = \hat{G}^\tb \cup \{ e\}\setminus \{r\}$ be the graph with $r = (u,v)$ removed and $e$ included.
Exactly one of the following events happens. 
\begin{enumerate}
    \item There exists an augmenting path $P$ from $u$ to $v$. 
    Then $r$ moves into the upper heap ($r \in U^\ta$), 
    $e$ becomes the root ($\hat{e}^\ta = r^{\ta} = e$), 
    and the matching is updated with the augmenting path, specifically
    $
    M^\ta = \Aug(M^\tb, P)
    $.
    \item There is no such augmenting path. Then $M^{\ta} = M^{\tb}$ and $e$ moves into the lower heap ($e \in L^{\ta}$), 
    and the root remains the same ($\hat{e}^\ta = r^{\ta} =r$).
\end{enumerate}
\end{lemma}

\begin{proof}
    Let $M' = M^\tb\setminus \{r\}$, and again $M'$ is a matching of size $n-1$ in $G'$ (hence it is not perfect), and the unmatched vertices are $u$ and $v$.

    Similar to the $M$-Event, if there exists an augmenting path $P$ from $u$ to $v$, augment $M'$ by replacing  $M'\cap P$ with $P\setminus M'$ to make a new matching $M''$.
    This increases $\abs{M'}$ by $1$; thus, $M''$ is a perfect matching. 
    Further, because $G' = \hat{G}^\tb \cup \{ e\}\setminus \{r\}$ and $e$ and $r$ have swapped places, we have that $G' = G_{\wt^\ta(e)}$.
    Moreover, $\wt^\ta(e) = \max\{\wt^\ta(e)\mid e\in M''\}$.
    Therefore $e = \hat{e}^\ta = r^\ta$, and $r$ gets moved down to $\uheap$ and becomes a child of $e$.

    Assume instead that there is no augmenting path in $G'$ from $u$ to $v$. 
    Then $M'$ is a maximum matching for $G'$ by Theorem \ref{thm:Berge}, however it is not perfect. 
    Therefore, there is no perfect matching for $G_{\wt^\ta(e)}$.
    However, $M^\tb$ is a perfect matching at time $\tb$ and $\hat{G}^\tb = G_{\wt^\ta(r)}\setminus\{e\}$, it still is a perfect matching at $\ta$ for $G_{\wt^\ta(r)}$.
    This is thus an internal event; we move $e$ to $\lheap$ as a child of $r$, and $r$ remains the bottleneck edge, i.e.~$\wt^\ta(r) = \max\{\wt^\ta(e)\mid e\in M''\}$.
\end{proof}

\subsection{Complexity and Performance Evaluation}
\label{sec:hanger}
The kinetic hourglass's complexity analysis largely follows that of the kinetic heap and kinetic hanger \cite{daFonseca2004Hanger}.
Recall the time complexities of the structures summarized in Table \ref{tbl:complexity} are obtained by multiplying the number of events by the time to process each event, $O(\log n)$.
This is the key difference between those data structures and the kinetic hourglass.

In the kinetic hourglass structure, we maintain two heaps or hangers at the same time.
For a bipartite graph $G$ of size $m$, we maintain the $m$ edges in the kinetic hourglass.
It is also the worst-case maximum number of elements in the max-heap $\lheap$ and the min-heap $\uheap$.
We therefore use $m = n$ in our analysis.
Within a single heap, we can think of this procedure as focusing on the cost of an edge at time $t$ as $c^t(e)$ only given for its time inside the heap. 
This means that when we evaluate the runtime for the heap (or hanger) structures, we can think of the function on each edge as being a curve segment, rather than defined for all time. 

The time complexities for the kinetic hourglass are summarized in Table \ref{tbl:hourglass} and we describe them further here.
The main change in time complexity results from the augmenting path search required at every external event. 
Note that in both the $L$- and $M$-events from Lemma \ref{lem:M} and \ref{lem:U}, exactly only one iteration of augmenting path search is required, so this takes $O(\abs{E}) = O(m)$ per iteration.
Therefore, we replace a $\log n$ factor with $m$.
When $\wt(e)$ is linear, the complexity of the kinetic hourglass is $O(m^2\sqrt{m}\log^{3/2}m)$ using heaps, and is $O(m^2\alpha(m)\log m)$ using hangers, since the arrangement of $m$ lines has complexity $O(m^2)$. 
If the weight function is non-linear, we are in the $s$-intersecting curve segments scenario.
The resulting runtimes are $O(m^3\log m)$ and $O(m\lambda_{s+2}(m)\log m)$, using heaps and hangers respectively.
The constant $s$ is determined by the behavior of the weight function. 

While this data structure is responsive, local, and compact, following directly from the analysis of kinetic heap and hanger, it fails to be efficient due to the linear time complexity required by each external event of the heap.
We conjecture that this can be improved via amortization analysis, which involves investigating the ratio between the number of internal and external events.
However, this is a non-trivial problem and has yet to be addressed in existing literature \cite{daFonseca2004Hanger}.

\begin{table}
    \centering
    \begin{tabular}{c|c|c}
        \textbf{Scenario} & \textbf{Kinetic Heap Hourglass} & \textbf{Kinetic Hanger Hourglass} \\
        \hline
         Line segments & $O(m^{5/2}\log^{1/2}m)$ & $O(m^2\alpha(m)\log m)$\\
        \hline
        $s$-intersecting curve segments  & $O(m^3\log m)$ & $O(m\lambda_{s+2}(m)\log m)$
    \end{tabular}
    \caption{Deterministic complexity of the kinetic heap hourglass and expected complexity of the kinetic hanger hourglass.}
    \label{tbl:hourglass}
\end{table}

\section{Kinetic Hourglass for Persistent Homology Transform}
\label{Sec:PHT}

In the following section, we apply our kinetic hourglass to compute the bottleneck distance between two persistent homology transforms \cite{TurnerPHT2014, CurryPHT2018, GhristPHT2018}, a specific case of persistence vineyard \cite{Morozov2006Vineyard}.
We first discuss the general PHT, then show how the vineyard structure enables the application of our kinetic hourglass framework.
In this section, we assume a basic knowledge of persistent homology; see \cite{DeyWang2017} for additional details.

\subsection{The Persistent Homology Transform}
Given a finite geometric simplicial complex $K$ in $\R^2$ with $|K|\subset \R^2$ denoting the geometric realization in the plane, every unit vector $\omega\in\S^1$ corresponds to a height function in direction $\omega$ defined as
\begin{equation*}
\begin{matrix}
    h_\omega:& |K| &\rightarrow&\R\\
    & x &\mapsto&\brak{x,v}
\end{matrix}
\end{equation*}
where $\brak{\cdot, \cdot}$ denotes the inner product. 
This induces a filtration $\{ h_\omega^{-1}(-\infty, a] \mid a \in \R \}$ of $|K|$, however we instead use the the combinatorial  representation given on the abstract complex $K$. 
The \emph{lower-star filtration} \cite{EdelsHarer2010} is defined on the abstract simplicial complex by 
\begin{equation*}
\begin{matrix}
    h_\omega:& K &\rightarrow&\R\\
    & \sigma &\mapsto&\max\{ h_v(u) \mid u \in \sigma\}. 
\end{matrix}
\end{equation*}
We abuse notation and write both functions as $h_\omega$ since the sublevelset $ h_\omega^{-1}(-\infty,a]$ has the same homotopy type in both cases, thus does not affect the persistent homology calculations. 
Note also that the combinatorial viewpoint means we can also define $h_\omega^{-1}(-\infty,a]$ as the full subcomplex of $K$ given by the vertex set $V_a = \{ v \in V \mid h_\omega(v) \leq a\}$. 

We say $K$ is generic if there are no two parallel and distinct lines each passing through pairs of vertices in $K$.
Then for a generic $K$, all vertices have distinct function values for $h_\omega$ except for $\omega$s that are perpendicular to any line connecting a pair of vertices, whether or not this is an edge.
Fix $\omega$ away from this set. 
Then the function induces a sorted order of the vertices by function value, and we denote the full subcomplex defined by the first $i$ vertices by in the order given by $h_\omega$.
Finally, the nested sequence of complexes 
\begin{equation*}
    \emptyset = K(\omega)_0\subseteq K(\omega)_1 \subseteq\dots\subseteq K(\omega)_n = K
\end{equation*} 
is the lower star filtration.
The resulting $k$-dimensional persistence diagram computed by filtering $K$ by the sub-level sets of $h_\omega$ is denoted $\Dgm_k(h^K_\omega)$.

In the context of this work, we are focusing on $0$-dimensional diagrams, denoted as $\Dgm(h^K_\omega)$ for simplicity.
Dimension 0 homology is entirely determined by the 0- and 1-dimensional simplices, so we can assume our input $K$ is an embedded graph. 
We assume that the input embedded graph is connected, so each $\Dgm(h^K_\omega)$ has exactly one feature that dies at time $t = \infty$.

\begin{definition}
\label{def:pht}
    The \emph{persistent homology transform} of $K\subset\R^2$ is the function 
    \begin{equation*}
    \begin{matrix}
        \PHT(K):&\S^1&\rightarrow &\mathcal{D}\\
        &\omega&\mapsto& \Dgm(h_\omega^K)
    \end{matrix}
    \end{equation*}
    where $h_\omega^K: K\rightarrow\R$, $h_\omega^K(x) = \brak{x, \omega}$ is the height function on $K$ in direction $\omega$, and $\mathcal{D}$ is the space of persistence diagrams.
\end{definition}

The PHT has been shown to be injective: for $K_1$, $K_2\subset\R^d$, if $\PHT(K_1) = \PHT(K_2)$, then $K_1 = K_2$ \cite{TurnerPHT2014, CurryPHT2018, GhristPHT2018}.  
This makes it a useful representation of shape in data analysis contexts. 
However, in those contexts, we wish to be able to compare two of these representations to encode their similarity. 
To this end, we are interested in comparing the  persistent homology transforms of two input complexes, $K_1$ and $K_2$, via an extension of the bottleneck distance using the framework of the kinetic hourglass.
\begin{definition}
\label{def:phtBottleneck}
    The \textit{integrated bottleneck distance} between $\PHT(K_1)$ and $\PHT(K_2)$ is defined as
    \[
    d_B(\PHT(K_1), \PHT(K_2)) = \int_0^{2\pi}d_B\left(\Dgm(h_\omega^{K_1}), \Dgm(h_\omega^{K_2})\right)d\omega.
    \]
\end{definition}

For a fixed direction $\omega$, the bottleneck distance between the two persistence diagrams can be formulated as a geometric matching problem as described in Section \ref{sec:BackgroundPersistence}.
Moreover, it is stable in the following sense. 
Assume that we have two geometric simplicial complexes $K_1$ and $K_2$ with the same underlying abstract complex. 
Then  we can think of the geometric embeddings as functions $f_1,f_2: K \to \R^2$ and so for a fixed $\omega \in \S^1$, we write $h_{i,\omega}:K \to \R$ for the respective height functions. 
We can use the fact that we have a lower star filtration along with the Cauchy-Schwartz inequality to see that 
$
\|h_{1,\omega} - h_{2,\omega}\|_\infty \leq \max_{v \in V} \|f_1(v) - f_2(v)\|_\infty
$. 
Then by the stability theorem \cite{CohenSteiner2007Stability}, the distance between the diagrams is bounded by 
\begin{equation*}
d_B(\Dgm(h_\omega^{K_1}), \Dgm(h_\omega ^{K_2})) \leq \max_{v \in V} \|f_1(v) - f_2(v)\|_\infty
\end{equation*}
for all $\omega\in\S^1$.
Integrating this inequality immediately gives the following result for the integrated bottleneck distance. 
\begin{proposition}
For two finite, connected geometric simplicial complexes with the same underlying abstract complex,        
\[
d_B(\PHT(K_1), \PHT(K_2))
\leq 2\pi \max_{v \in V} \|f_1(v) - f_2(v)\|_\infty.
\]
\end{proposition}

A choice of total order on the simplices of $K$ induces a pairing on the simplices of $K$, where for a pair $(\tau, \sigma)$, $\dim(\sigma) = \dim(\tau)+1$, and a $\dim(\tau)$-dimensional homology class was born with the inclusion of $\tau$ and dies with the inclusion of $\sigma$. 
Because the lower star filtration, the function values of these simplices are given by the unique maximum vertex $v_\sigma = \{ v \in \sigma \mid f_\omega(v) \geq f_\omega(u) \; \forall u \in \sigma)\}$ and so the function value when this simplex was added is given by $f_\omega(v_\sigma)$, and $v_\sigma$ does not change for nearby $\sigma$ that do not pass through one of the non-generic directions. 
For this reason, we can associate each finite point $(b,d)$ in the persistence diagram $\Dgm_k(h^K_\omega)$ to a vertex pair $(v_b, v_d)$ with $f_\omega(v_b) = b$ and $f_\omega(v_d) = d$. 
This pair is well-defined in the sense that given a point in the diagram, there is exactly one such pair. 
Further, because of the pairing on the full set of simplicies, every birth vertex appears in exactly one such pair, although here a death vertex could be associated to multiple points in the diagram. 
For any infinite points in the diagram of the form $(b,\infty)$, we likewise have a unique vertex $v_b$ again with $f_\omega(v_b) = b$.

Write a given filtration direction as $\omega = (\cos(\theta), \sin(\theta)) \in\S^1$ and fix  a point $x \in \Dgm(h_\omega^K)$ with birth vertex $v_a \in K$  located at coordinates  $(a_1, a_2) \in \R^2$ and death vertex $v_b \in K$ at coordinates $(b_1, b_2) \in \R^2$. 
This means the corresponding point in $x \in \Dgm(h_\omega^{K})$ has coordinates which we denote as 
\begin{equation*}
x(\omega) = (a_1\cos \theta + a_2 \sin v, b_1\cos v+ b_2\sin \theta).
\end{equation*}
Notice that for some interval of directions $I_x \subset \S^1$, the vertices $v_a$ and $v_b$ will remain paired, and so for that region, we will have the point $x(\omega')$ in the diagram for all directions $\omega' \in I_x$. 
Viewing this as a function of angle $\theta$, we observe that $x(\omega)$ is a parametrization of an ellipse,  meaning that the point $x(\omega)$ in the diagram will trace out a portion of an ellipse in the persistence diagram plane.
Further, the projection of this point to the diagonal has the form 
\begin{equation*}
x'(\omega) = \left( 
\tfrac{a_1 +b_1}{2} \cos \theta + \tfrac{a_2+b_2}{2} \sin \theta, 
\tfrac{a_1 +b_1}{2} \cos \theta + \tfrac{a_2+b_2}{2} \sin \theta 
\right),
\end{equation*}
which is also a parameterization of an ellipse, albeit a degenerate one.

\subsection{Kinetic Hourglass for PHT}
\label{sec:KDS-PHT}

In this section, we show how the kinetic hourglass data structure investigated in Sec.~\ref{Sec:KDS} can be applied to compute the exact distance between the 0-dimensional PHTs of two geometric simplicial complexes in $\R^2$, and provide an exact runtime analysis of the kinetic hourglass by explicitly determining the number of curve crossings in the flight plans for the edges in the bipartite graph.

\begin{remark}
While \cite{Arya2024} established the vine structure for star-shaped 
complexes with trivial monodromy, the kinetic hourglass algorithm 
requires only the \emph{local} vine structure, which is well-defined 
for any generic simplicial complex. The global monodromy affects 
whether vines return to themselves after a full rotation of $\theta$, 
but the algorithm processes events locally and computes the correct 
bottleneck distance at each direction.
\end{remark}

For the purposes of this section, we work with a simplicial complex $K$ that is generic, meaning there are no two parallel and distinct lines each passing through pairs of vertices in $K$.
For such a generic $K$, all vertices have distinct function values for $h_\omega$ except for $\omega$s that are perpendicular to any line connecting a pair of vertices.

For a fixed direction $\omega$ away from these non-generic directions, we can associate each finite point $(b,d)$ in the persistence diagram $\Dgm_k(h^K_\omega)$ to a vertex pair $(v_b, v_d)$ with $f_\omega(v_b) = b$ and $f_\omega(v_d) = d$. 
As the direction $\omega$ varies over $\S^1$, these points trace out continuous paths in the persistence diagram plane.
We call these paths \emph{vines}, following the vineyard terminology of \cite{Morozov2006Vineyard}.

Specifically, write a given filtration direction as $\omega = (\cos(\theta), \sin(\theta)) \in\S^1$ and fix a point $x \in \Dgm(h_\omega^K)$ with birth vertex $v_a \in K$ located at coordinates $(a_1, a_2) \in \R^2$ and death vertex $v_b \in K$ at coordinates $(b_1, b_2) \in \R^2$. 
The corresponding point in $x \in \Dgm(h_\omega^{K})$ has coordinates 
\begin{equation*}
x(\omega) = (a_1\cos \theta + a_2 \sin \theta, b_1\cos \theta + b_2\sin \theta).
\end{equation*}
For some interval of directions $I_x \subset \S^1$, the vertices $v_a$ and $v_b$ will remain paired, and so for that region, we will have the point $x(\omega')$ in the diagram for all directions $\omega' \in I_x$. 
Viewing this as a function of angle $\theta$, we observe that $x(\omega)$ traces out a portion of an ellipse in the persistence diagram plane.
Similarly, the projection of this point to the diagonal traces out a (degenerate) ellipse.

We call a vertex $v$ in a geometric simplicial complex $K$ an \emph{extremal vertex} if it is not in the convex hull of its neighbors, following \cite{Wang2024}. 
Write $\ext(V)$ for the set of vertices of $K$ which are extremal. 
By Lemma 2.5 of \cite{Wang2024}, a vertex in $K$ gives birth to a 0-dimensional class for some direction $\omega$ if and only if it is an extremal vertex. 
Thus, we can bound the number of vines $N$ by the number of extremal vertices: $N \leq |\ext(V)|$.

For an extremal vertex $v$, we can find the interval of directions $I_v = [\omega_1, \omega_2] \subset \S^1$ for which it gives birth.
The directions $\omega_1$ and $\omega_2$ are the normal vectors of the edges that form the largest angle incident to $v$.
We can write a piecewise-continuous path associated to vertex $v$ as $\mu_v: I_v \to \R^2; \, \omega \mapsto x(\omega)$ for the point $x(\omega)$ in the diagram with birth vertex $v$, where each piece is a portion of an ellipse.

\begin{lemma}
    For an extremal vertex $v$, if $\mu_v (\omega) \in \Delta$ then $\omega = \omega_1$ or $\omega = \omega_2$. 
\end{lemma}
\begin{proof}
    Let $\omega_1$, $\omega_2$ be the normal vectors of the external edges $(v, u_1)$, $(v, u_2)$ respectively, and
    $\mu_v(\omega) = (h_\omega(v), h_\omega(v'))$ where $v$ is the birth-causing vertex and $v'$ is the death causing vertex.
    Since $v$ is birth-causing, 
    all other neighbors of $v$ have higher height values along direction $\omega$: $\{h_\omega(u) > h_\omega(v) \mid u\in N_K(v), u\neq v'\}$.
    Because $\mu_v(\omega) \in \Delta$, then $h_\omega(v)=h_\omega(v')$. 
    Namely $\brak{v, \omega} = \brak{v', \omega}$.
    This implies that $\omega$ is a normal vector to the edge defined by $v$ and $v'$.
    Furthermore, no neighbor of $v$ lies in the half plane defined by $(v, v')$ and $-\omega$.
    Therefore, $(v, v')$ must be part of the convex hull formed by $v$ and $N_K(v)$.
    We thus conclude that $v'$ is either $u_1$ or $u_2$, and $\omega = \omega_1$ or $\omega = \omega_2$.
\end{proof}

Given two simplicial complexes $K_1$ and $K_2$, suppose the two PHTs have vines $\{\gamma_{1,i}\}_{i=1}^{n_1}$ and $\{\gamma_{2,j}\}_{j=1}^{n_2}$ respectively, and write the projection of the respective vine to the diagonal as $\gamma_{k,i}^\Delta$. 
We construct the complete bipartite graph $G$ with vertex set $U \cup V$ where $U$ and $V$ each have $n_1+n_2$ vertices, each associated to one of the vines from either vineyard.
In $U$ we have the vines $\{\gamma_{1,i}\}_i$ followed by the projections $\{\gamma_{2,j}^\Delta\}_j$; the reverse holds for $V$. 
Writing $\gamma_u$ for the vine associated to vertex $u$, the edge weights are given by
\begin{equation*}
    c(u,v) = 
    \begin{cases}
        \|\gamma_u(\omega) - \gamma_v(\omega)\|_\infty & \text{if } \gamma_u(\omega) \not \in \Delta \text{ or } \gamma_v(\omega) \not \in \Delta\\
        \|\gamma_u(\omega) - \Delta\|_\infty & \text{if } \gamma_v(\omega) \in \Delta \\
        \|\gamma_v(\omega) - \Delta\|_\infty & \text{if } \gamma_u(\omega) \in \Delta \\
        0 & \text{else}.
    \end{cases}
\end{equation*}
See Figure~\ref{fig:WeightFunc} for an example of the weight functions.

Given the above construction, let $n = \max\{\abs{\text{ext}^{K_1}(V)}, \abs{\text{ext}^{K_2}(V)}\}$. We have at most $2n$ nodes in each vertex set $U$ and $V$ and thus at most $4n^2$ edges in $G$.
Therefore, the kinetic heap hourglass maintains $4n^2$ elements, resulting in a runtime of $O((4n^2)^3\log (4n^2)) = O(n^6\log n)$.

\begin{figure}
    \centering
    \includegraphics[width=\linewidth]{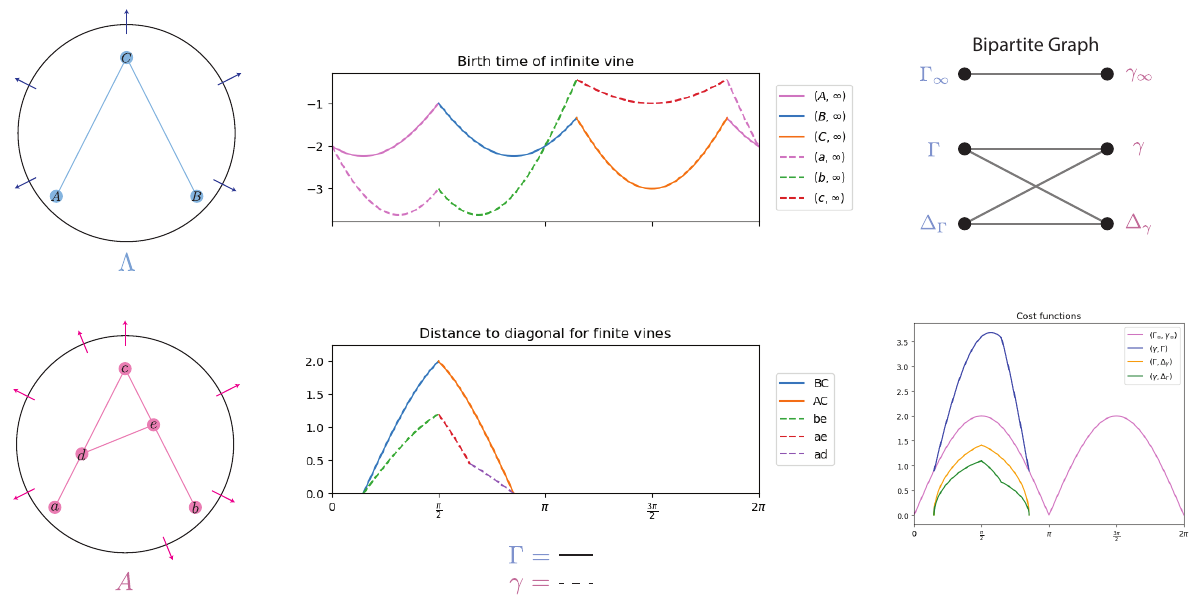}
    \caption{An example of the weight function $\wt$. The two figures in the middle visualize the behavior of the vines. On the top right is the bipartite graph representation, and on the bottom are the cost functions.}
    \label{fig:WeightFunc}
\end{figure}

\section{Conclusions}

In this paper, we introduced the \emph{kinetic hourglass}, a novel kinetic data structure designed for maintaining the bottleneck distance for graphs with continuously changing edge weights. 
The structure incorporates two kinetic priority queues, which can be either the deterministic kinetic heap or the randomized kinetic hanger. 
Both versions are straightforward to implement.

We applied this data structure to compute the integrated bottleneck distance between two persistent homology transforms of geometric simplicial complexes in $\R^2$, achieving $O(n^6 \log n)$ complexity.
This result has since been improved by Kerber and Wang \cite{KerberWang2025}, who achieve $\tilde{O}(n^5)$ for the integral objective and $\tilde{O}(n^3)$ for the maximum bottleneck distance in $\R^2$ using different techniques.

In the future, we hope to improve the runtime for this data structure. 
In particular, the augmenting path search requires $O(n)$ time, falling short of the efficiency goals in the kinetic data structure (KDS) framework. 
Moreover, when comparing PHTs with $n$ vertices, the kinetic hourglass holds $n^2$ elements, which can be computationally expensive. 
This method can immediately be extended to study the extended persistent homology transform (XPHT) to compare objects that have different underlying topologies \cite{TurnerXPHT2024}.
Further, the nature of this data structure means that it is not immediately extendable to the Wasserstein distance case, however, we would like to build a modified version that will work for that case.
Finally, since the kinetic hourglass data structure also has the potential to compare more general vineyards that extend beyond the PHT, it will be interesting in the future to find further applications.

\section*{Acknowledgments}
The work was supported in part by the National Science Foundation through grants CCF-2142713, CCF-2106578, and CCF-2107434.
EXW acknowledges support from the Austrian
Science Fund (FWF) grant number P~33765-N.